\documentclass{LMCS}

\def\dOi{9(3:21)2013}
\lmcsheading%
{\dOi}
{1--17}
{}
{}
{Oct.~20, 2012}
{Sep.~17, 2013}
{}

\usepackage[utf8]{inputenc}
\usepackage{amsmath}
\usepackage{amssymb}
\usepackage{graphicx}
\usepackage{pgf}
\usepackage{pgfarrows}
\usepackage{pgfnodes}
\usepackage{tikz}

\usepackage{proof}
\usepackage{verbatim}
\usepackage{float}
\usepackage{macros}

\usepackage{enumerate,hyperref}

\begin{document}

\title[Reachability under Contextual Locking]{Reachability under Contextual Locking\rsuper*}

\author[R.~Bonnet]{R\'emi Bonnet\rsuper a}
\address{{\lsuper a}LSV, ENS Cachan \& CNRS, France}
\email{remi.bonnet@cs.ox.ac.uk}

\author[R.~Chadha]{Rohit Chadha\rsuper b}
\address{{\lsuper{b,c,d}}University of Missouri, U.S.A.}
\email{chadhar@missouri.edu, madhu@illiniois.edu, vmahesh@illinois.edu}

\author[P.~Madhusudan]{P.~Madhusudan\rsuper c}
\address{\vspace{-6 pt}}

\author[M.~Viswanathan]{Mahesh Viswanathan\rsuper d}
\address{\vspace{-6 pt}}

\ACMCCS{[{\bf Theory of computation}]: Semantics and
  reasoning---Program reasoning---Program verification; [{\bf Software
      and its engineering}]: Software organization and
  properties---Software functional properties---Formal
  methods---Software verification}

\subjclass{F.3.1, D.2.4}
\keywords{Static analysis, Pushdown reachability, Locks, Reentrant locks} 
\titlecomment{{\lsuper*}An extended abstract of the paper appeared in
\cite{rch:pm:mv:12}.}

\maketitle

\begin{abstract}
The pairwise reachability problem for a multi-threaded program asks,
given control locations in two threads, whether they can be simultaneously reached in
an execution of the program. The problem is important for static
analysis and is used to detect statements that are concurrently enabled. This
problem is in general undecidable even when data is abstracted and
when the threads (with recursion) synchronize only using a finite set of locks. 
Popular programming paradigms that limit the lock usage patterns have been identified
under which the pairwise reachability problem becomes decidable. In
this paper, we consider a new natural programming paradigm, called contextual locking,
which ties the lock usage to calling patterns in each thread: we assume that locks are released in the same context that they were acquired and that every lock acquired by a thread in
a procedure call is released before the procedure returns. Our
main result is that the pairwise reachability problem is polynomial-time 
decidable for this new programming paradigm as well. The problem becomes undecidable if the locks are
reentrant; reentrant locking is a \emph{recursive locking} mechanism which allows a thread in a multi-threaded program 
to acquire the reentrant lock multiple times.  \end{abstract}
 \section{Introduction}
In static analysis of sequential programs~\cite{muchnick}, such as
control-flow analysis, data-flow analysis, points-to analysis, etc., 
the semantics of the program and the data that it manipulates is \emph{abstracted}, and the
analysis concentrates on computing fixed-points over a lattice using
the control-flow in the program. For instance, in flow-sensitive
context-sensitive points-to analysis, a finite partition of the heap
locations is identified, and the analysis keeps track of the set of
possibilities of which variables point may point to each heap-location
partition, propagating this information using the control-flow graph
of the program. In fact, several static analysis questions can be
formulated as reachability in a pushdown system that captures the
control-flow of the program (where the stack is required to model
recursion)~\cite{HorowitzRepsSagiv}.

In concurrent programs, abstracting control-flow is less obvious, due to the
various synchronization mechanisms used by threads to communicate and orchestrate
their computations.
One of the most basic questions is \emph{pairwise reachability}: given
two control locations $pc_1$ and $pc_2$ in two threads of a concurrent
program, are these two locations simultaneously reachable? This
problem is very basic to static analysis, as many analysis techniques
would, when processing $pc_1$, take into account the \emph{interference} of
concurrent statements, and hence would like to know if a location like
$pc_2$ is concurrently reachable. Data-races can also be formulated
using pairwise reachability, as it amounts to asking whether a
read/write to a location (or an abstract heap region) is concurrently
reachable with a write to the same location (or region). More
sophisticated verification techniques like deductive verification can
also utilize such an analysis. For instance, in an Owicki-Gries style
proof~\cite{og76} of a concurrent program, the invariant at $pc_1$
must be proved to be stable with respect to concurrent moves by the
environment, and hence knowing whether $pc_2$ is concurrently
reachable will help determine whether the statement at $pc_2$ need be
considered for stability.

Pairwise reachability of control locations is hence an important
problem. Given that individual threads may employ recursion, this
problem can be \emph{modeled} as reachability of \emph{multiple}
pushdown systems that synchronize using the synchronization constructs
in the concurrent program, such as locks, barriers, etc. However, it
turns out that even when synchronization is limited to using just
locks, pairwise reachability is
\emph{undecidable}~\cite{ram00}.   Consequently, recently, many natural
restrictions have been identified under which pairwise reachability is
decidable.

One restriction that yields a decidable pairwise reachability problem
is \emph{nested locking}~\cite{kig05,kg06}: if each thread performs
only nested locking (i.e. locks are released strictly in the reverse order
in which they are acquired), then pairwise reachability is known to be
decidable~\cite{kig05}.  The motivation for nested locking is that many
high-level locking constructs in programming languages naturally
impose nested locking. For instance the {\tt synchronize(o) \{\ldots
  \}} statement in Java acquires the lock associated with $o$, executes the
body, and releases the lock, and hence nested synchronized blocks
naturally model nested locking behaviors. The usefulness of the pairwise 
reachability problem was demonstrated in~\cite{kig05} where the above  decision procedure for
nested locking was used to find bugs in the Daisy file system.
 Nested locking has been
generalized to the paradigm of \emph{bounded lock chaining} for which
pairwise reachability has also been proved to be
decidable~\cite{kah09,kah11}. 

In this paper, we study a different restriction on locking, called
\emph{contextual locking}. A program satisfies contextual locking if
each thread, in every context, acquires new locks and releases all
these locks before returning from the context. Within the context,
there is \emph{no requirement} of how locks are acquired and released;
in particular, the program can acquire and release locks in a
non-nested fashion or have unbounded lock chains.

The motivation for contextual locking comes from the fact that this is
a very natural restriction. First, note that it's very natural for
programmers to release locks in the same context they were acquired;
this makes the acquire and release occur in the same syntactic code block,
which is a very simple way of managing lock acquisitions.

Secondly, contextual locking is very much encouraged by higher-level
locking constructs in programming languages.  For example, consider
the code fragment of a method, in Java~\cite{lea} shown in Figure
\ref{fig:javaexxample}.
\begin{figure}
\begin{verbatim}
   public void m() {
     synchronized(done) {
       ...
       synchronized(r) {
         ...
         while (done=0)
         try {
            done.wait();
         }
     ...
   }
\end{verbatim}
\caption{Synchronized blocks in Java}
\label{fig:javaexxample}
\end{figure}
The above code takes the lock associated with $\textit{done}$ followed
later by a lock associated with object $r$.  In order to proceed, it
wants $done$ to be equal to $1$ (a signal from a concurrent thread,
say, that it has finished some activity), and hence the thread waits
on $\textit{done}$, which releases the lock for $\textit{done}$,
allowing other threads to proceed. When some other thread issues a
\emph{notify}, this thread wakes up, reacquires the lock for $\textit{done}$,
and proceeds.

Notice that despite having synchronized blocks, the {\texttt wait()}
statement causes releases of locks in a non-nested fashion (as it
exhibits the sequence \emph{acq lock\_done; acq lock\_r; rel
  lock\_done; acq lock\_done; rel lock\_r; rel lock\_done}).
However, note that the code above does satisfy \emph{contextual
  locking}; the lock $m$ acquires are all released before the exit,
because of the synchronized-statements. Thus, we believe that
contextual locking is a natural restriction that is adhered to
in many programs.

The first result of this paper is that pairwise reachability is
decidable under the restriction of contextual locking. It is worth
pointing out that this result does not follow from the decidability
results for nested locking or bounded lock
chains~\cite{kig05,kah09}. Unlike nested locking and bounded lock
chains, contextual locking imposes no restrictions on the locking
patterns in the absence of recursive function calls; thus, programs
with contextual locking may not adhere to the nested locking or
bounded lock chains restrictions.  Second, the decidability of nested
locking and bounded lock chains relies on a non-trivial observation
that the number of context switches needed to reach a pair of states
is bounded by a value that is \emph{independent} of the size of the
programs.\footnote{This observation is implicit in the proofs of decidability in~ \cite{kig05,kah09}.} However, such a result of a bounded number of context
switches does not seem to hold for programs with contextual locking. Thus, the
proof techniques used to establish decidability are different as well.

We give a brief outline of the proof ideas
behind our decidability result. We observe that if a pair of states is
simultaneously reachable by some execution, then they are also
simultaneously reachable by what we call a \emph{well-bracketed
  computation}. A concurrent computation of two threads is not well-bracketed, if in the computation one process, say $\cP_0$, makes a
call which is followed by the other process ($\cP_1$) making a call,
but then $\cP_0$ returns from its call before $\cP_1$ does (but after
$\cP_1$ makes the call). We then observe that every well-bracketed
computation of a pair of recursive programs can simulated by a single
recursive program. Thus, decidability in polynomial time follows from
observations about reachability in pushdown systems~\cite{bem97}.

The second result of the paper concerns  \emph{reentrant  locks}. The standard mutex locks are blocking, \emph{i.e.}, if a
lock is held by some thread, then any attempt to acquire it by any thread (including the \emph{owning} thread) fails and the requesting thread blocks. 
Some programming languages such as Java support  (non-blocking)  \emph{reentrant  locks}.
Reentrant locks are recursive locks;  if a thread attempts to acquire a reentrant lock it already holds 
then the thread succeeds.  The lock becomes free only when the owning thread releases the lock as many times as
it acquires the lock.

We  consider the case of multi-threaded programs synchronizing through \emph{reentrant} locks and show
 that the pairwise reachability problem is undecidable for contextual reentrant locking (even for non-recursive programs). 
The undecidability result is obtained by a reduction from the emptiness problem of a $2$-counter machine. 
As the locking mechanism is reentrant,  a counter can be simulated by a lock. Counter increment and counter decrement can be
simulated by lock acquisition and lock release respectively. The ``zero-test'' in the reduction is simulated by 
synchronization between the  threads.

The rest of the paper is organized as follows. Section~\ref{sec:model}
introduces the model of concurrent pushdown systems communicating
via locks and presents its semantics. Our main decidability result
is presented in Section~\ref{sec:main}. The undecidability result for contextual reentrant locking is presented in Section~\ref{sec:reentrant}. Conclusions are presented in
Section~\ref{sec:conclusions}.

{\bf Note:} An extended abstract of the paper  co-authored by Rohit Chadha, P. Madhusudan and Mahesh Viswanathan appeared in
\cite{rch:pm:mv:12} and contains the decidability
 result for pairwise reachability under contextual (non-reentrant) locking.
The undecidability result for contextual reentrant locking was obtained subsequently by R\'emi Bonnet and Rohit Chadha,
and has not been presented elsewhere.

\section{Model}
\label{sec:model}
The set of natural numbers shall be denoted by $\Nats.$ The set of functions from a set $A$ to $B$ shall be denoted by $B^A.$ Given a function $f\in B^A$, the function $f |_{a\mapsto b}$ shall be the unique function $g$ defined as follows:
$g(a)=b$ and for all $a'\ne a$, $g(a')= f(a').$ If $\bar a = (a_1,\ldots,a_n)\in A_1\times\cdots\times A_n$ then $\pi_i(a)=a_i$ for each $1\leq i\leq n.$

\subsection{Pushdown Systems.} For static analysis,  recursive programs are usually modeled as pushdown systems.
Since we are interested in modeling threads in concurrent programs we will also need to model locks for communication between threads.
Formally,

\begin{definition}
Given a finite set $\Locks$, a pushdown system (PDS) $\cP$ using $\Locks$
is a tuple $(\Q,\Gamma,qs,\delta)$ where

\begin{iteMize}{$\bullet$}
\item $\Q$ is a  finite set of control states.
\item $\Gamma$ is a finite stack alphabet.
\item $qs$ is the initial state.
\item $\delta=\delta_\state \union \delta_\push \union \delta_\pop \union \delta_\acq \union \delta_\rel$ is a finite set of
   transitions where
   \begin{iteMize}{$-$}
   \item $\delta_\state \subseteq \Q  \times \Q $.
   \item $\delta_\push \subseteq \Q  \times (\Q \times \Gamma)$.
   \item $\delta_\pop \subseteq (\Q \times \Gamma)\times \Q$.
   \item $\delta_\acq\subseteq \Q  \times (\Q\times \Locks)$.
   \item $\delta_\rel \subseteq  (\Q \times \Locks)\times \Q  $.

   \end{iteMize}
 \end{iteMize}
 \end{definition}
For a PDS $\cP$, the semantics is defined as a transition system.
The configuration of a PDS $\cP$ is the product of the set of control states $\Q$ and the stack which is modeled
as a word over the stack alphabet $\Gamma.$  For a thread $\cP$ using $\Locks,$ we have to keep track of the locks being held by $\cP.$
 Thus
  the set of configurations  of  $\cP$ using $\Locks$ is  $\Conf_\cP= \Q\times \Gamma^*\times 2^\Locks$
 where  $2^\Locks$ is the powerset of $\Locks.$

 Furthermore, the transition relation is no longer just
 a relation between configurations but  a binary relation on $2^\Locks\times \Conf_\cP $ since
 the thread now {\it executes} in an {\it environment}, namely, the free locks (i.e., locks not being held by any other thread).
 Formally,

 \begin{definition}
 A PDS $\cP=(\Q,\Gamma,\qs,\delta)$ using $\Locks$ gives  a labeled
 transition relation $\longrightarrow_\cP \subseteq (2^\Locks \times (\Q\times \Gamma^* \times 2^\Locks))\times\Labels\times (2^\Locks \times(\Q\times \Gamma^* \times 2^\Locks))$ where $\Labels=\{\state,\push,\pop\}\union\{\acq(l),
 \rel(l)\ |\ l \in \Locks\}$ and
 $\longrightarrow_\cP$ is defined as follows.
 \begin{iteMize}{$\bullet$}
 \item $\free:(q,w,\held)\  {\stackrel{\state}{\longrightarrow}}_\cP\ \free:(q',w,\held)$ if $(q,q')\in \delta_\state.$
 \item $\free:(q,w,\held) \stackrel{\push}{\longrightarrow}_\cP \free:(q',wa,\held)$ if $(q,(q',a))\in \delta_\push.$
 \item $\free:(q,wa,\held) {\stackrel{\pop}\longrightarrow}_\cP\ \free:(q',w,\held)$ if $((q,a),q')\in \delta_\pop.$
 \item $\free:(q,w,\held)\ {\stackrel{\acq(l)}{\longrightarrow}}_\cP\ \free\setminus{\{l\}}:(q',w,\held\union \{l\})$ if $(q,(q',l))\in \delta_\acq$ and $l\in \free$.
 \item $\free:(q,w,\held)\ {\stackrel{\rel(l)}\longrightarrow}_\cP\ \free\union{\{l\}}:(q',w,\held\setminus \{l\})$ if  $((q,l),q')\in \delta_\rel$ and
 $l\in \held$.
 \end{iteMize}
 \end{definition}

\subsection{Multi-pushdown systems} Concurrent programs are modeled as multi-pushdown systems.
For our paper, we assume that  threads in a concurrent program communicate only through locks which leads
us to the following definition.
\begin{definition}
Given a finite set $\Locks$, a $n$-pushdown system ($n$-PDS) $\cCP$ communicating via $\Locks$
is a tuple $(\cP_1,\ldots, \cP_n)$ where each $\cP_i$ is a PDS using $\Locks.$
\end{definition}

Given a $n$-PDS $\cCP$,
we will assume that the set of control states and the stack symbols of the threads are mutually disjoint.
\begin{definition}
The semantics of a $n$-PDS $\cCP=(\cP_1,\ldots, \cP_n)$ communicating via $\Locks$ is given as a labeled transition system $T=(S,s_0,\longrightarrow)$ where
\begin{iteMize}{$\bullet$}
\item $S$ is said to be the set of configurations of $\cCP$ and  is the set $(\Q_1\times \Gamma_1^* \times 2^\Locks)\times \cdots \times(\Q_n\times \Gamma_n^* \times 2^\Locks),$ where $\Q_i$ is the set of states of $\cP_i$ and $\Gamma_i$ is the stack alphabet of $\cP_i.$
\item $s_0$ is the initial configuration and is $((qs_1,\epsilon, \emptyset),\cdots, (qs_m,\epsilon,\emptyset))$ where $qs_i$ is the initial state of $\cP_i.$

\item The set of labels on the transitions  is $\Labels \times \{1,\ldots, n\}$ where $\Labels = \{\state,\push,\pop\}\union\{
\acq(l),\rel(l)\ |\ l \in \Locks\}.$
The labeled transition relation $\stackrel{(\lambda,i)}{\longrightarrow}$ is defined as follows
 $$ ((q_1,w_1, \held_1),\cdots (q_n,w_n,\held_n))\stackrel{(\lambda,i)}{\longrightarrow}  ((q'_1,w'_1, \held'_1),\cdots (q'_n,w'_n,\held'_n))$$ iff
      $$ \Locks \setminus \union_{1\leq r\leq n}{\held_r}: (q_i,w_i,\held_i) \stackrel{\lambda}\longrightarrow_{\cP_i} \Locks \setminus \union_{1\leq r\leq n}{\held'_r}:(q_i',w_i',\held_i')$$
      and for all $j\ne i$,
      $q_j=q'_j,$ $w_j=w_j'$ and $\held_j=\held_j'$.
\end{iteMize}
\end{definition}

\begin{notation}

Given a configuration $s= ((q_1,w_1, \held_1),\cdots, (q_n,w_n,\held_n))$ of a $n$-PDS $\cCP$,
we say that $\Conf_i(s)=(q_i,w_i,\held_i), \ControlSt_i(s)=q_i, \Stack_i(s)= w_i, \LockSet_i(s)=\held_i$
and $\StackHeight_i(s)=|w_i|$, the length of $w_i.$
\end{notation}

\subsection{Computations.}
 A {\it computation} of  the $n$-PDS $\cCP$, $\Comp$, is a sequence $s_0\stackrel{(\lambda_1,i_1)}{\longrightarrow}\cdots \stackrel{(\lambda_m,i_m)}{\longrightarrow}s_m$
 such that $s_0$ is the initial configuration of
$\cCP.$ The {\it label of the computation} {\Comp}, denoted {$\Label(\Comp)$}, is said to be the word $(\lambda_1,i_1)\cdots(\lambda_m,i_m).$
The
transition $s_j\stackrel{(\push,i)}\longrightarrow s_{j+1}$ is said to be a {\it procedure call by thread $i$}. Similarly, we can define {\it procedure return}, {\it internal action},
{\it acquisition of lock $l$} and {\it release of lock $l$} by thread $i.$
A procedure return $s_{j}\stackrel{(\pop,i)}\longrightarrow s_{j+1}$   is said to
{\it match} a procedure call $s_{\ell}\stackrel{(\push,i)}\longrightarrow s_{\ell+1}$ iff $\ell<j,$
$\StackHeight_i(s_\ell)=\StackHeight_i(s_{j+1})$ and for all $\ell+1 \leq p\leq j,$
$\StackHeight_i(s_{\ell+1})\leq \StackHeight_i(s_p).$




\begin{example}
\label{exam:example1} Consider the two-threaded program showed in
Figure \ref{fig:example1}. For sake of convenience, we only show the relevant actions of the programs.
Figure \ref{fig:compexam} shows  computations whose labels are as follows:
$$\begin{array}{ll}
\Label(\Comp1)=& (\push,0) (\acq(\textrm{l1}),0) (\push,1) (\acq(\textrm{l2}),0) (\rel(\textrm{l1}),0)
        (\acq(\textrm{l1}),1)\\
      & \hspace*{4cm}(\rel(\textrm{l2}),0) (\pop,0) (\rel(\textrm{l1}),1) (\pop,1)
\end{array}$$

and
$$ \begin{array}{ll}
\Label(\Comp2)= &(\push,0) (\acq(\textrm{l1}),0) (\push,1) (\acq(\textrm{l2}),0) (\rel(\textrm{l1}),0)
        (\acq(\textrm{l1}),1)\\
        &\hspace*{4cm} (\rel(\textrm{l1}),1) (\pop,1)(\rel(\textrm{l2}),0) (\pop,0).
\end{array}$$
respectively.
\begin{figure}[ht]
\begin{minipage}[b]{0.5\textwidth}
\begin{verbatim}
   int a(){
      acq l1; 
      acq l2;
      if (..) then{
         ....
         rel l2;
         ....
         rel l1;
          };
      else{
          .....
          rel l1
          .....
          rel l2
          };
         return i;
   };

   public void P0() {
     n=a();
   }
   \end{verbatim}

\end{minipage}
\begin{minipage}[b]{0.5\textwidth}
\begin{verbatim}
  int b(){
      acq l1;
      rel l1;
      return j;
   };

   public void P1() {
     m=b();
   }
   \end{verbatim}

\end{minipage}

\caption{A $2$-threaded program with threads P0 and P1}
\label{fig:example1}
\end{figure}

\begin{figure}
\begin{center}
\begin{tikzpicture}
\input{comp.tkz}
\end{tikzpicture}
\end{center}
\caption{Computations $\Comp 1$ and $\Comp 2$. Transitions of $P0$ are
  shown as solid edges while transition of $P1$ are shown as dashed
  edges; hence the process ids are dropped from the label of
  transitions. Matching calls and returns are shown with dotted
  edges.}
\label{fig:compexam}
\end{figure}

\end{example}

\subsection{Contextual locking}
In this paper, we are considering the pairwise reachability problem when the threads follow  {\it contextual locking}. Informally, this means that--
\begin{iteMize}{$\bullet$}
\item every lock acquired by a thread in a
procedure call must be released before the corresponding return is executed, and
\item the locks held by a thread just before a procedure call is executed are not released during the
execution of the procedure.

\end{iteMize}
Formally,
 \begin{definition}
A thread $i$ in a $n$-PDS $\cCP=(\cP_1,\ldots, \cP_n)$ is said to follow contextual locking if whenever $s_\ell \stackrel {(\push,i)}\longrightarrow s_{\ell+1}$
and $s_{j}\stackrel  {(\pop,i)} \longrightarrow s_{j+1}$ are matching procedure call and return  along a computation
$s_0   \stackrel{(\lambda_1,i)}\longrightarrow s_1  \cdots \stackrel {(\lambda_m,i)}\longrightarrow s_m,$
we have that
$$\LockSet_i(s_\ell)= \LockSet_i(s_j) \mbox{ and } \mbox{ for all } \ell\leq r \leq j.\ \LockSet_i(s_\ell) \subseteq \LockSet_i(s_r).$$

 \end{definition}

\begin{example}
\label{exam:example2}
Consider the $2$-threaded program shown in Figure \ref{fig:example1}. Both the threads P0 and
P1 follow contextual locking. The program P2 in Figure \ref{fig:example2} does not follow contextual locking.
\begin{figure}[ht]
\begin{verbatim}

   int a(){
      acq l1;
      rel l2;
      return i;
   };
   public void P2(){
   acq l2;
   n=a();
   rel l1;
   }
\end{verbatim}

\caption{A program that does not follow contextual locking.}
\label{fig:example2}
\end{figure}

\end{example}
\begin{example}
\label{exam:example3}
Consider the $2$-threaded program in Figure \ref{fig:example3}. The
two threads P3 and P4 follow contextual locking as there is no
recursion! However, the two threads do not follow either the
discipline of nested locking~\cite{kig05} or of bounded lock
chaining~\cite{kah09}. Hence, algorithms of \cite{kig05,kah09} cannot
be used to decide the pairwise reachability question for this
program. Notice that the computations of this pair of threads require
an unbounded number of context switches as the two threads proceed in
lock-step fashion. The locking pattern exhibited by these threads can
present in any program with contextual locking as long as this pattern
is within a single calling context (and not across calling
contexts). Such locking patterns when used in a non-contextual fashion
form the crux of undecidability proofs for multi-threaded programs
synchronizing via locks~\cite{kig05}.

\begin{figure}[ht]
\begin{minipage}{0.5\textwidth}
\begin{verbatim}

   public void P3(){
    acq l1;
    while (true){
       acq l2;
       rel l1;
       acq l3;
       rel l2;
       acq l1;
       rel l3;
    }
    }
   \end{verbatim}

\end{minipage}
\begin{minipage}{0.5\textwidth}
\begin{verbatim}
  public void P4(){
    acq l3;
    while (true){
       acq l1;
       rel l3;
       acq l2;
       rel l1;
       acq l3;
       rel l2;
    }
    }
   \end{verbatim}

\end{minipage}

\caption{A $2$-threaded program with unbounded lock chains}
\label{fig:example3}
\end{figure}

\end{example}


\section{Pairwise reachability}
\label{sec:main}
The pairwise reachability problem for a multi-threaded program
asks whether two given states in two threads can be simultaneously
reached in an execution of the program. Formally,

\begin{definition}
Given a $n$-PDS $\cCP=(\cP_1,\ldots, \cP_n)$ communicating via $\Locks,$   indices $1\leq i,j\leq n$ with $i\ne j$, and control states $q_i$ and
$q_j$ of threads $\cP_i$ and $\cP_j$ respectively, let $Reach(\cCP,q_i,q_j)$ denote the predicate that there is a computation $s_0\longrightarrow \cdots \longrightarrow s_m$ of $\cCP$ such that $\ControlSt_{i}(s_m)=q_i$ and $\ControlSt_{j}(s_m)=q_j.$ The
pairwise control state reachability problem  asks if $Reach(\cCP,q_i,q_j)$  is true.
\end{definition}

The pairwise reachability problem for multi-threaded programs
communicating via locks was first studied in~\cite{ram00}, where it
was shown to be undecidable. Later, Kahlon et. al.~\cite{kig05} showed
that when the locking pattern is restricted the pairwise reachability
problem is decidable.  In this paper, we will show that the problem is
decidable for multi-threaded programs in which each thread follows
contextual locking. Before we show this result, note that it suffices
to consider programs with only two threads \cite{kig05}.
\begin{proposition}
Given a $n$-PDS $\cCP=(\cP_1,\ldots, \cP_n) $ communicating via
$\Locks,$ indices $1\leq i,j\leq n$ with $i\ne j$ and control states
$q_i$ and $q_j$ of $\cP_i$ and $\cP_j$ respectively, let $\cCP_{i,j}$
be the $2$-PDS $(\cP_i,\cP_j)$ communicating via $\Locks.$ Then $Reach(\cCP,q_i,q_j)$ iff
$Reach(\cCP_{i,j},q_i,q_j).$
\end{proposition}
Thus, for the rest of the section, we will only consider $2$-PDS.

\subsection{Well-bracketed computations}
The key concept in the proof of decidability is
the concept of well-bracketed computations, defined below.

\begin{definition}
Let $\cCP=(\cP_0,\cP_1)$ be a $2$-PDS via $\Locks$ and let $\Comp=s_0\stackrel{(\lambda_1,i_1)}{\longrightarrow}\cdots \stackrel{(\lambda_m,i_m)}{\longrightarrow}s_m$ be a computation of $\cCP$. $\Comp$ is said to be  \emph{non-well-bracketed} if there exist $0\leq \ell_1<\ell_2<\ell_3<m$ and $i\in \{0,1\}$ such that
\begin{iteMize}{$\bullet$}
\item $s_{\ell_1}\stackrel{(\push,i)}{\longrightarrow}s_{\ell_1+1}$  and $s_{\ell_3}\stackrel{(\pop,i)}{\longrightarrow}s_{\ell_3+1}$ are matching call and returns of
$\cP_i$, and
\item $s_{\ell_2}\stackrel{(\push,i)}{\longrightarrow}s_{\ell_2+1}$  is a procedure call of thread $\cP_{1-i}$ whose matching return either occurs after $\ell_3+1$ or does not occur at all.
\end{iteMize}
Furthermore,
the triple $(\ell_1,\ell_2,\ell_3)$ is said to be a {\it witness} of non-well-bracketing of $\Comp$.

$\Comp$ is said to be \emph{well-bracketed} if it is not non-well-bracketed.
\end{definition}
\begin{example}
\label{exam:example4}
Recall the $2$-threaded program from Example~\ref{exam:example1} shown in Figure~\ref{fig:example1}.
The computation $\Comp1$ (see Figure~\ref{fig:compexam}) is non-well-bracketed, while the computation $\Comp2$
(see Figure~\ref{fig:compexam}) is well-bracketed. On the other hand, all the computations of the $2$-threaded program
in  Example~\ref{exam:example3} (see Figure~\ref{fig:example3}) are well-bracketed as the two threads are
non-recursive.
\end{example}
 If there is a computation that simultaneously reaches
control states $p\in \cP_0$ and $q\in \cP_1$   then there is a well-bracketed computation that simultaneously reaches
$p$ and $q$:

\begin{lemma}
Let
 $\cCP=(\cP_0,\cP_1)$ be a $2$-PDS communicating via $\Locks$ such that each thread follows contextual locking. Given control states $p\in \cP_0$ and $ q\in \cP_1$, we have that
$Reach(\cCP,p,q)$ iff there is a  well-bracketed computation $s^{wb}_0 \longrightarrow \cdots \longrightarrow s^{wb}_r$ of $\cCP$
such that $\ControlSt_0(s^{wb}_r)=p$ and $\ControlSt_1(s^{wb}_r)=q.$

\end{lemma}
\begin{proof}
Let $\Comp_{nwb}=s_0\stackrel{(\lambda_1,i_1)}{\longrightarrow}\cdots \stackrel{(\lambda_m,i_m)}{\longrightarrow}s_m$ be a non-well-bracketed computation that simultaneously reaches $p$ and $q$.
Let $\ell_\mn$ be the smallest $\ell_1$ such that there is a witness $(\ell_1,\ell_2,\ell_3)$ of non-well-bracketing of $\Comp_{nwb}$. Observe now that it suffices to show that there is another computation $\Comp_{mod}$ of the same length as
$\Comp_{nwb}$ that simultaneously reaches $p$ and $q$ and
\begin{iteMize}{$\bullet$}
\item either $\Comp_{mod}$ is well-bracketed,
\item or if $\Comp_{mod}$ is non-well-bracketed then for each witness $(\ell_1',\ell_2',\ell_3')$ of
  non-well-bracketing of $\Comp_{mod},$ it must be the case $\ell_1'>\ell_\mn.$
\end{iteMize}

\noindent We show how to construct $\Comp_{mod}.$ Observe first that any witness $(\ell_\mn, \ell_2,\ell_3)$ of non-well-bracketing of $\Comp_{nwb}$ must necessarily agree in the third component $\ell_3.$ Let $\ell_\rt$
denote this component.  Let $\ell_{\sm}$ be the smallest
$\ell_2$ such that $(\ell_\mn, \ell_2,\ell_\rt)$ is a witness of non-well-bracketing of $\Comp_{mod}.$
Thus, the transition $s_{\ell_\mn}\longrightarrow s_{\ell_\mn+1}$ and $s_{\ell_\rt}\longrightarrow s_{\ell_\rt+1}$
are matching procedure call and return of some thread $\cP_r$ while the transition $s_{\ell_\sm}\longrightarrow s_{\ell_\sm+1}$ is a procedure call of $c'$ by thread $\cP_{1-r}$ whose return happens only after $\ell_\rt.$ Without loss of generality, we can assume that $r=0.$

Let $u,(\push,0),v_1,(\push,1),v_2,(\pop,0)$ and $w$ be such that
$\Label(\Comp_{nwb})= u (\push,0) v_1(\push,1)$ $v_2(\pop,0)w$.  and
length of $u$ is $\ell_\mn$, of $u(\push,0)v_1$ is $\ell_\sm$ and
of $u(\push,0)v_1(\push,1)v_2$ is $\ell_\rt.$ Thus, $(\push,0)$ and
$(\pop,0)$ are matching call and return of thread $\cP_0$ and
$(\push,1)$ is a call of the thread $\cP_1$ whose return does not
happen in $v_2$.

We construct $\Comp_{mod}$ as follows. Intuitively, $\Comp_{mod}$ is
obtained by ``rearranging" the sequence
$\Label(\Comp_{nwb})=u(\push,0)v_1(\push,1)v_2(\pop,0)w$ as
follows. Let $v_2|0$ and $v_2|1$ denote all the ``actions" of thread
$\cP_0$ and $\cP_1$ respectively in $v_2.$ Then $\Comp_{mod}$ is
obtained by rearranging $u(\push,0)v_1(\push,1)v_2(\pop,0)w$ to
$u(\push,0)v_1(v_2|0)(\pop,0)(\push,1)(v_2|1)w.$ This is shown in
Figure \ref{fig:proof}.
\begin{figure}
\begin{center}
\begin{tikzpicture}
\input{proof-comp.tkz}
\end{tikzpicture}
\end{center}
\caption{Computations $\Comp_{nwb}$ and $\Comp_{mod}$. Transitions of
  $\cP_0$ are shown as solid edges and transitions of $\cP_1$ are
  shown as dashed edges; hence process ids are dropped from the label
  of transitions. Matching calls and returns are shown with dotted
  edges. Observe that all calls of $\cP_1$ in $v_1$ have matching
  returns within $v_1$.}
\label{fig:proof}
\end{figure}

The fact that if $\Comp_{mod}$ is non-well-bracketed then there is no
witness $(\ell'_1,\ell'_2,\ell'_3)$ of non-well-bracketing with
$\ell'_1\leq \ell_\mn$ will follow from the following observations on
$\Label(\Comp_{nwb}).$

\begin{iteMize}{$\dagger\dagger$}
\item[$\dagger$] $v_1$ cannot contain any returns of $\cP_1$ which
  have a matching call that occurs in $u$ (by construction of
  $\ell_\mn$).
\item[$\dagger\dagger$] All calls of $\cP_1$ in $v_1$ must return
  either in $v_1$ or after $c'$ is returned. But the latter is not
  possible (by construction of $\ell_\sm$). Thus, all calls of $\cP_1$
  in $v_1$ must return in $v_1.$
\end{iteMize}

Formally, $\Comp_{mod}$ is constructed as follows. We fix some
notation.  For each $0\leq j\leq m,$ let $\Conf^j_0=\Conf_0(s_j)$ and
$\Conf^j_1=\Conf_1(s_j).$ Thus $s_j=(\Conf^j_0,\Conf^j_1).$
\begin{enumerate}[(1)]

\item The first ${\ell_\sm}$ transitions of $\Comp_{mod}$ are the
  same as $\Comp_{nwb}$, i.e., initially
  $\Comp_{mod}=s_0\longrightarrow \cdots \longrightarrow
  s_{{\ell_\sm}}.$
  
\item Consider the sequence of transitions
  $s_{{\ell_{\sm}}}\stackrel{(\lambda_{\sm+1},i_{\sm+1})}\longrightarrow
  \cdots \stackrel{(\lambda_{\rt}+1,i_{\rt}+1)}\longrightarrow
  s_{\ell_{\rt+1}}$ in $\Comp$. Let $k$ be the number of transitions of $\cP_0$
  in this sequence and let ${{\ell_\sm}} \leq j_1< \cdots< j_k \leq
  \ell_{\rt}$ be the indices such that
  $s_{j_n}\stackrel{(\lambda_{j_n+1},0)}\longrightarrow
  s_{j_{n}+1}$. Note that it must be the case that for each $1\leq n< k,$
 $$\Conf^{\ell_{\sm} }_{0}=\Conf_0^{j_1},\  \Conf^{j_n+1}_0=\Conf^{j_{n+1}}_0 \mbox{ and } \Conf^{j_k+1}_0=\Conf^{\rt+1}_0.$$

 For $1\leq n\leq k$,  let $$s'_{{\ell_\sm+n}}=(\Conf_0^{j_{n}+1},\Conf_1^{\ell_\sm}).$$
 Observe now that, thanks to contextual locking, the set of locks held by $\cP_1$ in $\Conf_1^{\ell_\sm}$ is a subset of the locks held
 by $\cP_{1}$ in $\Conf_1^{\ell_{j_n}}$ for each $1\leq n\leq k.$ Thus we can extend $\Comp_{mod}$
 by applying the $k$  transitions of $\cP_0$  used to obtain $s_{j_n}\longrightarrow s_{j_n+1}$ in $\Comp_{nwb}$.
 In other words, $\Comp_{mod} $ is now
 $$s_0\longrightarrow \cdots \longrightarrow s_{{\ell_\sm}}\stackrel{(\lambda_{j_1+1},0)}{\longrightarrow}s'_{{\ell_\sm+1}}\cdots\stackrel{(\lambda_{j_k+1},0)}{\longrightarrow}s'_{\ell_\sm+k}.$$

     Note that $s'_{\ell_\sm+k}=(\Conf_0^{\rt+1},\Conf_1^{\ell_\sm}).$
     Thus the set of locks held by $\cP_0$ in  $s'_{{\ell_\sm}+k}$ is exactly the set of locks held by $\cP_0$ at
        $\Conf_0^{\ell_\mn}.$

     \item
     Consider the sequence of transitions $s_{{\ell_{\sm}}}\stackrel{(\lambda_{\sm+1},i_{\sm+1})}\longrightarrow \cdots \stackrel{(\lambda_{\rt}+1,i_{\rt}+1)}\longrightarrow s_{\ell_{\rt+1}}$ in $\Comp$. Let $t$ be the number of transitions of
 $\cP_1$ in  this sequence and let
 ${{\ell_\sm}} \leq j_1< \cdots< j_t \leq \ell_{\rt}$ be the indices such that
  $s_{j_n}\stackrel{(\lambda_{j_n+1},1)}\longrightarrow s_{j_{n}+1}$. Note that it must be the case that for each $1\leq n< t,$
 $$\Conf^{j_1}_{1}=\Conf^{\ell_\sm}_1,\ \Conf^{j_n+1}_1=\Conf^{j_{n+1}}_1 \mbox{ and } \Conf^{j_t+1}_1=\Conf^{\rt+1}_1.$$

 For  $1\leq n\leq t,$ let $$s'_{\ell_\sm+k+n}=(\Conf_0^{\rt+1}, \Conf_1^{j_n+1}).$$ 
 Observe now that, thanks to contextual locking, the set of locks held by $\cP_0$ in $\Conf_0^{\ell_\rt+1}$ is
  exactly the set of locks held by $\cP_0$ at
        $\Conf_0^{\ell_\mn}$ and the latter is  a subset of the locks held
 by $\cP_{0}$ in $\Conf_0^{\ell_{j_n}}$ for each $1\leq n\leq t.$
 Thus we can extend $\Comp_{mod}$
 by applying the $t$  transitions of $\cP_1$  used to obtain $s_{j_n}\longrightarrow s_{j_n+1}$ in $\Comp_{nwb}$.
 In other words, $\Comp_{mod} $ is now
 $$s_0\longrightarrow \cdots \longrightarrow s'_{{\ell_\sm+k}}\stackrel{(\lambda_{j_1+1},1)}{\longrightarrow}s'_{{\ell_\sm+k+1}}\cdots\stackrel{(\lambda_{j_t+1},1)}{\longrightarrow}s'_{\ell_\sm+k+t}.$$

       Observe now that the extended $\Comp_{mod}$ is a sequence of $\rt+1$ transitions  and that the final configuration
       of $\Comp_{mod}$, $s'_{{\ell_\sm+k+t}}=(\Conf^{\rt+1}_{0},\Conf^{\rt+1}_1)$
       is exactly the configuration $s_{\rt+1}.$
      \item Thus, now we can extend $\Comp_{mod}$ as
      $$s_0\longrightarrow \cdots \longrightarrow s'_{\ell_\sm+k+t}=s_{\rt+1}\stackrel{(\lambda_{\rt+2},i_{\rt+2})}\longrightarrow \cdots \stackrel{(\lambda_m,i_{m})}\longrightarrow  s_m.$$
        Clearly, $\Comp_{mod}$ has the same length as $\Comp_{nwb}$ and simultaneously reaches $p$ and $q.$    There is also no  witness $(\ell'_1,\ell'_2,\ell'_3)$ of non-well-bracketing in $\Comp_{mod}$ with
$\ell'_1\leq \ell_\mn.$  The lemma follows. \qedhere  \end{enumerate} 
\end{proof}
\subsection{Deciding   pairwise reachability}
We are ready to show that the problem of checking pairwise reachability is decidable.

\begin{theorem} There is an algorithm that
given a $2$-threaded program $\cCP=(\cP_0,\cP_1)$ communicating via $\Locks$ such that $\cP_0$ and $\cP_1$ follow contextual locking, and control states $p$ and $q$
of $\cP_0$ and $\cP_1$ respectively decides if $Reach(\cP,p,q)$ is true or not. Furthermore, if $m$ and $n$ are the sizes of
the  programs $\cP_0$ and $\cP_1$ and $\ell$ the number of elements of $\Locks,$ then this algorithm has a running time of $2^{O(\ell)}O((mn)^3).$

\end{theorem}
\begin{proof}
The main idea behind the algorithm is to construct a single PDS $\cP_{comb}=(\Q,\Gamma,\qs,\delta)$ which simulates all the well-bracketed
computations of  $\cCP.$  $\cP_{comb}$ simulates a well-bracketed computation as follows.
The set of control states of $\cP_{comb}$ is the product of control states of $\cP_0$ and $\cP_1.$
The single stack of $\cP_{comb}$ keeps track of the stacks of $\cP_0$ and $\cP_1$: it
is the sequence of those calls of the well-bracketed computation which have not been returned.
 Furthermore, if the stack of $\cP_{comb}$
is $w$ then  the stack of $\cP_0$ is the projection of $w$ onto the stack symbols of $\cP_0$ and the
stack of $\cP_1$ is the projection  of $w$ onto the stack symbols of $\cP_1.$ Thus, the top of the stack is the most recent unreturned call and if it belongs to $\cP_i,$ well-bracketing ensures
that no previous unreturned call is returned without returning this call.

Formally, $\cP_{comb}=(\Q,\Gamma,\qs,\delta)$ is defined as follows. Let $\cP_0=(\Q_0,\Gamma_0,\qs_0,\delta_0)$ and
$\cP_1=(\Q_1,\Gamma_1,\qs_1,\delta_1).$ Without loss of generality,  assume that $\Q_0\cap\Q_1=\emptyset$
and $\Gamma_0\cap\Gamma_1=\emptyset.$

\begin{iteMize}{$\bullet$}
\item   The set of states $\Q$ is $(\Q_0 \times 2^\Locks) \times (\Q_1\times 2^\Locks).$

  \item $\Gamma=\Gamma_0\union \Gamma_1.$ 
    \item $\qs=((\qs_0,\emptyset),(\qs_1,\emptyset)).$
     \item $\delta$ consists of three sets $\delta_\state,\delta_\push$ and $\delta_\pop$ which simulate the
             internal actions, procedure calls and returns, and lock acquisitions and releases of the threads
              as follows. We explain here only the simulation of actions of $\cP_0$ (the simulation of actions of $\cP_1$ is similar).
                 \begin{iteMize}{$-$}
                 \item {\it Internal actions.} If $(q_0,q_0')$ is an internal action of $\cP_0,$ then for each $\held_0,\held_1\in 2^\Locks$ and $q_1\in \Q_1 $
                  $$(((q_0, \held_0),(q_1,\held_1)), ((q'_0,\held_0),(q_1,\held_1))) \in \delta_\state.$$
                  \item {\it Lock acquisitions.} Lock acquisitions are also modeled by $\delta_\state.$ If $(q_0,(q_0',l))$ is a lock
                    acquisition action of thread $\cP_0,$
                   then for each  $\held_0,\held_1\in 2^\Locks$ and $q_1\in \Q_1,$
                  $$(((q_0, \held_0),(q_1,\held_1)), ((q'_0,\held_0\union \{l\}),(q_1,\held_1))) \in \delta_\state \mbox{ if } l \notin \held_0 \union \held_1.$$
                   \item {\it Lock releases.} Lock releases are also modeled by $\delta_\state.$ If $((q_0,l),q_0')$ is a lock
                    release action of thread $\cP_0,$
                   then for each  $\held_0,\held_1\in 2^\Locks$ and $q_1\in \Q_1,$
                  $$(((q_0, \held_0),(q_1,\held_1)), ((q'_0,\held_0\setminus \{l\}),(q_1,\held_1))) \in \delta_\state \mbox{ if } l \in \held_0.$$
                    \item {\it Procedure Calls.}     Procedure calls are modeled  by $\delta_\push.$ If $(q_0,(q_0',a))$
                    is a procedure call of thread $\cP_0$ then for each
                    $\held_0,\held_1\in 2^\Locks$ and $q_1\in \Q_1,$
                    $$(((q_0, \held_0),(q_1,\held_1)), (((q'_0,\held_0),(q_1,\held_1)),a)) \in \delta_\push.$$

                   \item {\it Procedure Returns.}     Procedure returns are modeled  by $\delta_\pop.$ If $((q_0,a),q_0')$
                    is a procedure return of thread $\cP_0$ then for each
                    $\held_0,\held_1\in 2^\Locks$ and $q_1\in \Q_1, $
                    $$( ((  (q_0, \held_0),(q_1,\held_1)),a), ((q'_0,\held_0),(q_1,\held_1))) \in \delta_\pop.$$

                                  \end{iteMize}

\end{iteMize}

\noindent It is easy to see that $(p,q)$ is reachable in $\cCP$ by a
well-bracketed computation iff there is a computation of $\cP_{comb}$
which reaches $((p,\held_0),(q,\held_1))$ for some $\held_0,\held_1\in
2^\Locks.$ The complexity of the results follows from the observations
in~\cite{bem97} and the size of $\cP_{comb}$. \end{proof}

\section{Reentrant locking}
\label{sec:reentrant}
We now turn our attention to reentrant locking. Recall that a reentrant lock is a recursive lock which allows the thread owning a lock
to acquire the lock multiple times; the owning thread must release the lock an equal number of times before another thread can acquire
the lock.  Thus, the  set of configurations of a  thread  $\cP$ using \emph{reentrant locks} is the set $ \Q\times \Gamma^*\times \Nats^\Locks$
(and not the set  $\Q\times \Gamma^*\times 2^\Locks$ as in non-reentrant locks).  Intuitively, the elements of a configuration $(q,w,\held)$ of $\cP$ now have the following meaning: $q$ is the control state of $\cP,$ $w$ the contents of the stack and $\held:\Locks\to \Nats$ is a function that tells the number of times each lock has been acquired by $\cP.$ The
semantics of a PDS $\cP$ using reentrant $\Locks$ is formally defined as:\footnote{The definition of PDS $\cP$ using reentrant locks $\cP$ is exactly the definition of PDS $\cP$ using locks $\cP$ (see Section \ref{sec:model}).} 
 \begin{definition}
 A PDS $\cP=(\Q,\Gamma,\qs,\delta)$ using reentrant $\Locks$ gives  a labeled
 transition relation $\longrightarrow_\cP \subseteq (2^\Locks \times (\Q\times \Gamma^* \times \Nats^\Locks))\times\Labels\times (2^\Locks \times(\Q\times \Gamma^* \times  \Nats^\Locks))$ where $\Labels=\{\state,\push,\pop\}\union\{\acq(l),
 \rel(l)\ |\ l \in \Locks\}$ and
 $\longrightarrow_\cP$ is defined as follows.
 \begin{iteMize}{$\bullet$}
 \item $\free:(q,w,\held)\  {\stackrel{\state}{\longrightarrow}}_\cP\ \free:(q',w,\held)$ if $(q,q')\in \delta_\state.$
 \item $\free:(q,w,\held) \stackrel{\push}{\longrightarrow}_\cP \free:(q',wa,\held)$ if $(q,(q',a))\in \delta_\push.$
 \item $\free:(q,wa,\held) {\stackrel{\pop}\longrightarrow}_\cP\ \free:(q',w,\held)$ if $((q,a),q')\in \delta_\pop.$
 \item $\free:(q,w,\held)\ {\stackrel{\acq(l)}{\longrightarrow}}_\cP\ \free\setminus{\{l\}}:(q',w,\held|_{l\mapsto \held(l)+1})$ if $(q,(q',l))\in \delta_\acq$ and either $l\in \free$ or $\held(l)>0$.
 \item $\free:(q,w,\held)\ {\stackrel{\rel(l)}\longrightarrow}_\cP\ \free:(q',w,\held|_{l\mapsto \held(l)-1})$ if  $((q,l),q')\in \delta_\rel$ and
 $ \held(l)>1$.

 \item $\free:(q,w,\held)\ {\stackrel{\rel(l)}\longrightarrow}_\cP\ \free\union{\{l\}}:(q',w,\held|_{l\mapsto 0})$ if  $((q,l),q')\in \delta_\rel$ and
 $ \held(l)=1$.
 \end{iteMize}
 \end{definition}
 
The semantics of $n$-PDS $\cCP=(\cP_1,\ldots, \cP_n)$ communicating via reentrant $\Locks$ is given as a
transition system on  $ (\Q_1\times \Gamma_1^* \times \Nats^\Locks)\times \cdots \times(\Q_n\times \Gamma_n^* \times \Nats^\Locks)$
where $\Q_i$ and $\Gamma_i$ are the set of states and the stack alphabet  of process $\cP_i$ respectively. Formally,
\begin{definition}
The semantics of a $n$-PDS $\cCP=(\cP_1,\ldots, \cP_n)$ communicating via reentrant $\Locks$  is given as a labeled transition system $T=(S,s_0,\longrightarrow)$ where
\begin{iteMize}{$\bullet$}
\item $S$ is said to be the set of configurations of $\cCP$ and  is the set $ (\Q_1\times \Gamma_1^* \times \Nats^\Locks)\times \cdots \times(\Q_n\times \Gamma_n^* \times \Nats^\Locks),$ where $\Q_i$ is the set of  states of $\cP_i$ and $\Gamma_i$ is the stack alphabet of $\cP_i.$
\item The initial configuration $s_0$ is $(qs_1,$ $ \epsilon, \overline{0}), \cdots, (qs_m,\epsilon,\overline{0})$ where $qs_i$ is the initial local state of $\cP_i$ and $\overline{0}\in \Nats^\Locks$ is the function which takes the value $0$ for each 
$l\in \Locks.$

\item The set of labels on the transitions  is $\Labels \times \{1,\ldots, n\}$ where $\Labels = \{\state,\push,\pop\}\union\{
\acq(l),\rel(l)\ |\ l \in \Locks\}.$
The labeled transition relation $\stackrel{(\lambda,i)}{\longrightarrow}$ is defined as follows.
$$(q_1,w_1, \held_1),\cdots, (q_n,w_n,\held_n))\stackrel{(\lambda,i)}{\longrightarrow}
 ((q'_1,w'_1, \held'_1),\cdots, (q'_n,w'_n,\held'_n)) $$
iff  for all $j\ne i$,
      $q_j=q'_j,$ $w_j=w_j'$ and $\held_j=\held_j'$ and  
      $$\begin{array}{l} \Locks \setminus \union_{1\leq r\leq n} \set{l\st {\held_r(l)>0}}: ((q,q_i),w_i,\held_i) \stackrel{\lambda}\longrightarrow_{\cP_i}\\ 
      \hspace*{3.5cm}\Locks \setminus \union_{1\leq r\leq n} \set{l\st{\held'_r(l)>0}}:((q,q_i'),w_i',\held_i').\end{array}$$
\end{iteMize}
\end{definition}

\begin{notation}
Given a configuration $s= ((q_1,w_1, \held_1),\cdots, (q_n,w_n,\held_n))$ of a $n$-PDS $\cCP$ using reentrant $\Locks$,
we say that $\LockHeld_i(s)=\held_i.$
\end{notation}

\subsection{Contextual locking}
We now adapt contextual locking to reentrant locks. Informally, contextual reentrant locking means that--
\begin{iteMize}{$\bullet$}
\item each instance of a lock acquired by a thread in a
procedure call must be released before the corresponding return is executed, and
\item the instances of locks held by a thread just before a procedure call is executed are not released during the
execution of the procedure.
\end{iteMize}
Formally,
 \begin{definition}
A thread $i$ in a $n$-PDS $\cCP=(\cP_1,\ldots, \cP_n)$ is said to follow contextual locking if whenever $s_\ell \stackrel {(\push,i)}\longrightarrow s_{\ell+1}$
and $s_{j}\stackrel  {(\pop,i)} \longrightarrow s_{j+1}$ are matching procedure call and return  along a computation
$s_0   \stackrel{(\lambda_1,i)}\longrightarrow s_1  \cdots \stackrel {(\lambda_m,i)}\longrightarrow s_m,$
we have that
$$
\LockHeld_i(s_\ell)= \LockHeld_i(s_{j+1}) \mbox{ and for all } \ell\leq r \leq j.\ \LockHeld_i(s_\ell) \leq \LockHeld_i(s_r). $$
 \end{definition}

\subsection{Pairwise reachability problem}
We are ready to show that the pairwise reachability problem becomes undecidable if we have reentrant locks. 

\begin{lemma}
\label{thm:contextundecid}
The following problem is undecidable:
Given a $2$-PDS $\cCP=(\cP_1, \cP_2)$ communicating only via reentrant locks $\Locks$  s.t. $\cP_1$ and $\cP_2$ follow contextual locking, and control states $q_1$ and
$q_2$ of  threads $\cP_1$ and $\cP_2,$  check if $Reach(\cCP,q_1,q_2)$ is  true. The problem continues to remain undecidable even when $P_1$ and $P_2$ are finite state systems, \emph{i.e.}, $P_1$ and  $P_2$ do not use their stack during any computation.
\end{lemma}
\begin{proof}
We shall show a reduction from the halting problem of a two-counter machine (on empty input)  to the pairwise control state reachability problem. 

A two-counter machine $\cM$ is a tuple $(Q,q_s,q_f, \Delta)$:  $Q$ is a finite set of states, $q_s\in\Q$ is the initial state, $q_f\in Q$ is the final state and $\Delta$ is a tuple $(\Delta_{state},\{\Delta_{inc_i},\Delta_{dec_i}, \Delta_{z_i}\}_{i=1,2})$ where
 $\Delta_{state} \subseteq Q\times Q$ is the set of state transitions of $\cM$   for each $i=1,2$,  $\Delta_{inc_i}\subseteq  Q\times Q$ is the set  of increment transitions of
the  counter $i$, $\Delta_{dec_i}\subseteq  Q\times Q$ is the set of decrement transitions of
the  counter $i$, and $\Delta_{z_i}\subseteq  Q\times Q$ is the set of zero-tests of the  counter $i$.
A configuration of $\cM$ is a triple 
$(q,c_1,c_2)$ where $q$ is the ``current control state", and $c_1\in\Nats$ and $c_2\in \Nats$ are the values of the counters $1$ and $2$ respectively.  
$\cM$ is said to \emph{halt}  if there is a computation starting from $(q_s,0,0)$ and reaching a configuration $(q_f,c_1,c_2)$ for some $c_1,c_2\in\Nats.$
The halting problem asks if $\cM$ halts.

Fix a two-counter machine $\cM=(Q,q_s,q_f, \Delta).$ We will construct a  finite set $\Locks$, a $2$-PDS $\cCP=(\cP_1, \cP_2)$ communicating via reentrant $\Locks$ and states $q_1,q_2$ of $\cP_1$ and $\cP_2$ respectively such that $\cM$ halts iff $Reach(\cCP,q_1,q_2)$ is true. 

The set $\Locks$ will be the set $\set{h, h', r_1,r_2,l_1,l_2,t_1,t_2}.$ 
The $2$-PDS $\cCP$ will not have any recursion, \emph{i.e.}, its stack alphabet  will be empty. $\cCP$ will constructed in two steps. In the first step, we construct a $2$-PDS $\cCP^0=(\cP^0_1,\cP^0_2)$ communicating via reentrant   $\Locks$, a configuration $init=((p,\epsilon, \held_1),(q,\epsilon,\held_2))$ of $\cCP^0$ and states $q_1$ and $q_2$ of $\cP^0_1$ and $\cP^0_2$ respectively such that there is a computation of $\cCP^0$ starting from $init$ and ending in a state $((q_1,\epsilon,\held_1'),(q_2,\epsilon,\held_2'))$ for some $\held_1',\held_2'$ iff $\cM$ halts. In the next step of the construction, we will show how to ``set up'' the appropriate values of $\held_1$ and $\held_2$ initially.

We construct $\cCP^0$ as follows. The locks $h$ and $h'$ will not play a role in the construction of $\cCP^0$, but will
be used later to initialize $\held_1$ and $\held_2$.
 Intuitively,  the program $\cCP^0$  simulates a computation of $\cM$.
The simulation is mostly carried out by the thread $\cP^0_1$. $\cP^0_1$  maintains the control state
 of $\cM$ as well as the values of the two counters  using the locks $l_1$ and $l_2$: 
 the contents of the counter $i$ are stored as number of times  the lock $l_i$ is acquired by $\cP^0_1.$ 
 An increment of the counter $i$ is simulated by an acquisition of the lock $l_i$ and a decrement of $c_i$ is simulated by release of the lock $l_i.$ A state transition of $\cM$ is modeled by an internal action of $\cP^0_1.$

 A zero-test of  counter $i$ is carried by synchronization of threads $\cP^0_1$ and $\cP^0_2$ and involves  the locks 
 $r_i$ and $t_i.$ The zero test on counter $i$ is carried out as follows.  
 For the zero test to be carried out, $\cP^0_2$, must be in a ``ready" state $q_\ast.$  
 Before a zero-test on counter $i$ is carried out,  
  thread $\cP^0_1$ holds one instance of lock $r_i$ and thread $\cP^0_2$ holds one instance of lock $t_i.$ To carry out the zero-test, $\cP^0_1$ carries out the following sequence of lock acquisitions and releases:
  $$\acq(t_i)\rel(r_i)\acq(l_i)\rel(t_i)\acq(r_i)\rel(l_i) $$ and
 $\cP^0_2$ carries out the following sequence:
 $$\acq(l_i)\rel(t_i)\acq(r_i)\rel(l_i)\acq(t_i)\rel(r_i). $$

 Intuitively, the zero-test ``commences" by $\cP^0_2$ trying to acquire lock $l_i.$ If the lock acquisition succeeds then $\cP^0_2$ ``learns" that counter $i$ is indeed 
zero. $\cP^0_2$ then releases the lock $t_i$ in order to ``inform" $\cP^0_1$ that the counter $i$ is $0$ and ``waits" for 
thread $\cP^0_1$ to update its state. To test that counter $i$ is indeed $0$, thread $\cP^0_1$ tries to acquire
lock $t_i.$ If it succeeds in acquiring the lock then thread $\cP^0_1$ learns that counter $i$ is indeed $0.$ It updates its state, releases lock $r_i$ to  ``inform"  thread $\cP^0_2$ that it has updated the state. Thread $\cP^0_2$ learns that $\cP^0_1$  has updated its state by acquiring lock $r_i.$ Now, $\cP^0_2$  releases lock $l_i$ so that it can ``tell" $\cP^0_1$ to release the lock $t_i.$
Thread $\cP^0_1$ acquires lock $l_i;$ releases lock $t_i$ and waits for lock $r_i$ to be released. $\cP^0_2$ can now acquire  the lock $t_i$. After acquiring $t_i,$  $\cP^0_2$  releases lock $r_i$ and transitions back to the state $q_\ast.$ Thread $\cP^0_1$ acquires lock $r_i$ and releases lock $l_i.$ 
Now, we have that $\cP^0_1$ holds lock $r_i$ and $\cP^0_2$ holds lock $t_i$ as before and nobody holds $l_i$. Now $\cP^0_1$ can continue
simulating $\cM$.

Let $\held_1(h)=\held_1(r_1)=\held_1(r_2)=1 $ and  $\held_1(h')=\held_1(t_1)=\held_1(t_2)=\held_1(l_1)=\held_1(l_2)=0.$ Let
 $\held_2(h)=\held_2(r_1)=\held_2(r_2)=\held'_2(l_1)=\held_2'(l_2)=0,$ $\held_2(h')=\held_2(t_1)=\held_2(t_2)=1.$ 
Let $init=((q_s,\epsilon,\held_1),(q_\ast,\epsilon,\held_2)).$
The construction of
$\cP^0$ ensures that there is a computation of $\cCP^0$ starting from $init$ and ending in a state $((q_f,\epsilon,\held_1'),(q_\ast,\epsilon,\held_2'))$ for some $\held_1',\held_2'$ iff $\cM$ halts.

Now, $\cCP=(\cP_1,\cP_2)$ is constructed from $\cCP^0$ as follows. The thread $\cP_1$ initially performs the following sequence of lock acquisitions and releases:
$$\acq(h')\acq(r_1)\acq(r_2)\acq(h)\rel(h'); $$ 
makes a transition to $q_s$ and starts behaving like thread $\cP^0_1.$
The thread $\cP_2$ initially performs the following sequence of lock acquisitions and releases:
$$\acq(h)\acq(t_1)\acq(t_2)\rel(h)\acq(h'); $$
makes a transition to $q_\ast$ and starts behaving like thread $\cP^0_2.$ 

The construction ensures that $Reach(\cCP,q_f,q_\ast)$ is true iff $\cM$ halts. 
The formal construction of $\cCP$ has been carried out in the Appendix.
\end{proof}

\section{Conclusions}
\label{sec:conclusions}

The paper investigates the problem of pairwise reachability of
multi-threaded programs communicating using only locks. We identified
a new restriction on locking patterns, called contextual locking,
which requires threads to release locks in the same calling context in
which they were acquired. Contextual locking appears to be a natural
restriction adhered to by many programs in practice. The main result
of the paper is that the problem of pairwise reachability is decidable
in polynomial time for programs in which the locking scheme is
contextual. Therefore, in addition to being a natural restriction to
follow, contextual locking may also be more amenable to 
practical analysis.  We observe that these results do not follow from
results in~\cite{kig05,kg06,kah09,kah11} as there are programs
with contextual locking that do not adhere to the nested locking
principle or the bounded lock chaining principle. The proof principles
underlying the decidability results are also different. Our results can also be mildly extended to handling programs that release locks a bounded stack-depth away from when they were acquired
   (for example, to handle procedures that call a function that acquires a lock, and calls another to release it before it returns).

There are a few open problems immediately motivated by the results on contextual locking in
this paper. First, decidability of model checking with respect to
fragments of LTL under the contextual locking restriction remains
open. Next, while our paper establishes the decidability of pairwise
reachability, it is open if the problem of checking if 3 (or more) threads
simultaneously reach given local states is decidable for programs with
contextual locking. Finally, from a practical standpoint, one would
like to develop analysis algorithms that avoid to construct the
cross-product of the two programs to check pairwise reachability.

We also considered the case of reentrant locking mechanism and established that the pairwise reachability under
contextual reentrant locking is undecidable. The status of the pairwise reachability problem for the case when the
locks are nested (and not necessarily contextual) is open. This appears to be a very difficult problem. Our reasons for
believing this is that the problem of checking control state reachability in a PDS with \emph{one} counter and no zero tests
can be reduced to the problem of checking pairwise reachability problem  in a $2$-threaded program communicating via a single (and hence nested) reentrant lock.   
The latter is a long standing open problem.  

For a more complete account for multi-threaded programs, other synchronization
primitives such as thread creation and barriers should be taken into account. Combining
lock-based approaches such as ours with techniques for other primitives is left to
future investigation.

\subsection{Acknowledgements.}
P. Madhusudan was supported in part by NSF Career Award 0747041. Mahesh Viswanathan was 
supported in part by NSF CNS 1016791 and NSF CCF 1016989. Rohit Chadha was at LSV, ENS Cachan \& INRIA
during the time research was carried out.

\bibliographystyle{plain}
\bibliography{refs}
\appendix
\section{Construction of $\cCP$ in the proof of Theorem~\ref{thm:contextundecid}}

We carry out the construction of $\cCP$  in the proof of Theorem~\ref{thm:contextundecid}. 

Recall that  $\cM=(Q,q_s,q_f, \Delta)$ is  a two  counter machine and $\Delta$ is the tuple $(\Delta_{state},\{\Delta_{inc_i},$ $\Delta_{dec_i}, \Delta_{z_i}\}_{i=1,2}).$ $\cP$ is a $2$-PDS communicating via reenterant 
$\Locks=\set{h, h', r_1,r_2,l_1,l_2,t_1,t_2}.$ Recall also that $\cCP$ is constructed in two steps. First a $2$-PDS 
$\cCP=(\cP^0_1, \cP^0_2)$ is constructed and then extended to $\cCP.$  

Formally,  the set of states of $\cP^0_1$ is  $Q^1=Q\, \union\, (Q\times\set{1,2,3,4,5} \times \set{1,2}).$ Intuitively, the state
$(q,j,i)\in Q\times\set{1,2,3,4,5} \times \set{1,2}$ means that the counter $i$ is being tested for zero, the thread $\cP^0_1$ has completed its $j^{th}$ step in the test    and $q$ will be the resulting state after the zero test is completed.   
The transitions $\delta^1$
of  $\cP^0_1$ 
are defined as follows. The set of internal actions of $\cP^0_1$, $\delta_\state,$ is the set $\delta_{state}.$ 
The set of lock acquisitions $\delta_\acq$ is the union of sets $\delta_{\acq_1}$ and $\delta_{\acq_2}$ where
$\delta_{\acq_i}$ is the set
$
\set{(q^\dagger,(q^\ddagger,l_i))\st (q^\dagger,q^\ddagger)\in \delta_{inc_i}} \union \set{(q^\dagger,((q^\ddagger,1,i),t_i)) \st
 (q^\dagger,q^\ddagger)\in \delta_{z_i}}
 \union \set {((q,2,i), ((q,3,i),l_i) ) \st q\in Q} \union \set {((q,4,i), ((q,5,i),r_i) ) \st q\in Q}.$ The set of lock releases $\delta_\rel$ is the union of sets $\delta_{\rel_1}$ and $\delta_{\rel_2}$ where
$\delta_{\rel_i}$ is the set
$
\set{((q^\dagger,l_i),(q^\ddagger))\st (q^\dagger,q^\ddagger)\in \delta_{dec_i}} \union \set{(((q,1,i),r_i), (q,2,i) )\st q\in Q} \union \set{(((q,3,i),t_i), (q,4,i) )\st q\in Q} \union \set{(((q,5,i),l_i), q )\st q\in Q}.$

The set of states of $\cP^0_2$ is the set $\Q^2=\set{q_\ast} \union (\set{0,1,2,3,4,5}\times \set{1,2})$
where $q_\ast$ is a new state. Intuitively, the state $q_\ast$ means that $\cP^0_2$ is ready to test a counter.
The state $(j,i)$ means that the counter $i$ is being tested for zero and the thread $\cP^0_2$ has completed its $j^{th}$ step in the test.  The set of transitions $\delta'$
of  $\cP^0_2$ consists of only lock acquisition transitions and lock release transitions.     The set of lock acquisitions $\delta'_\acq$ is the union of sets $\delta'_{\acq_1}$ and $\delta'_{\acq_2}$, where
$\delta'_{\acq_i}=\set{(q_\ast,((1,i),l_i))\, ,\, ((2,i),((3,i),r_i))\, ,\, ((4,i),((5,i),t_i))}.$ The set of lock releases $\delta'_\rel$ is the union of sets $\delta'_{\rel_1}$ and $\delta'_{\rel_2}$, where
$\delta'_{\rel_i}=\set{(((1,i),t_i), (2,i))\, , \, ( ((3,i),l_i),(4,i) )\, ,$ $ \,(((5,i), r_i ), q_\ast )}.$

Now the $2$-PDS $\cCP=(\cP_1,\cP_2)$ is constructed as follows. For $\cP_1,$  the set of states is $Q_1=Q^1\union\set{q_0,q_1,q_2,q_3,q_4}$ where $q_0,q_1,q_2,q_3,q_4$ are new states. In addition to $\delta^1,$ the set of
transitions of $\cP_1$
 contains the  lock acquisition transitions 
$(q_0,(q_1,h'))$, $ (q_1,(q_2,r_1))$, $(q_2,(q_3,r_2))$, $(q_3,(q_4,h)) $ and the lock release transition $((q_4,h'),q_s).$ The state $q_0$ is the initial state of $\cP_1.$

For $\cP_2,$ the set of states is $Q_2=Q^2\union\set{q'_0,q'_1,q'_2,q'_3,q'_4}$ where $q'_0,q'_1,q'_2,q'_3,q'_4$ are new states. In addition to $\delta^2,$ the set of
transitions of $\cP_2$
contains the  lock acquisition transitions $(q'_0,(q'_1,h))$, $(q'_1,(q'_2,t_1))$, $(q'_2,(q'_3,t_2))$,  $(q'_4,(q'_\ast,h')) $ and the lock release transition $\set{((q'_3,h),q'_4)}.$ The state $q'_0$ is the initial state of $\cP_2.$

\end{document}